\documentclass[letterpaper, 10 pt, conference]{IEEEtran}

\usepackage{amsthm}
\usepackage{cite}
\usepackage{amsmath}
\usepackage{mathtools}
\usepackage{xcolor}
\usepackage{verbatim}
\usepackage[pdftex]{graphicx}
\usepackage{subcaption}
\usepackage{amsfonts}
\usepackage{algorithm}
\usepackage{algpseudocode} 
\DeclareMathOperator{\diag}{diag}
\DeclareMathOperator{\Diag}{Diag}
\DeclareMathOperator{\tr}{tr}
\DeclareMathOperator{\colspace}{colspace}

\newtheorem*{definition}{Definition}
\newtheorem*{theorem}{Theorem}

\begin{document}

\title{\LARGE \bf A Moving-target Cyber-Attack Detection Strategy for Large-scale Power Systems using Dynamic Clustering}

\author{Ana Jevti\'c and Marija Ili\'c%
	\thanks{Ana Jevti\'c and Marija Ili\'c are with the Laboratory for Information and Decision Systems (LIDS) at Massachusetts Institute of Technology, 
		Cambridge, MA 02139 USA
		{\tt\small ajevtic@mit.edu, ilic@mit.edu}}%
	\thanks{This material is based upon work supported by the Department of Energy under Award Number DE-OE0000779.}%
}

\maketitle

\begin{abstract}
	 In recent years, cyber-security of power systems has become a growing concern. To protect power systems from malicious adversaries, advanced defense strategies that exploit sophisticated detection algorithms are required. Motivated by this, in this paper we introduce an active defense method based on dynamic clustering. Our detection strategy uses a moving-target approach where information about the system's varying operating point is first used to cluster measurements according to their transfer function characteristics that change over time. Then, detection is carried out through series of similarity checks between measurements within the same cluster. The proposed method is effective in detecting cyber-attacks even when the attacker has extensive knowledge of the system parameters, model and detection policy at some point in time. The effectiveness of our proposed detection algorithm is demonstrated through a numerical example on the IEEE 24-bus power system.
\end{abstract}

\section{Introduction}
Cyber-security of power systems has become a very important topic of research in the recent years. Events like the 2015 cyber-attack on the Ukrainian power grid \cite{Ukraine2015}, and recently reported cyber incident that disrupted grid operation in the western US \cite{westUS} have led to the increase in awareness of the problem of securing critical infrastructure, such as the power grid, transportation systems, gas and water networks, etc. The large scale of power systems, diversity and complexity of its components, and, more recently, its exposure to the public via smart devices with Internet connectivity, are some of the factors that make cyber-security of power systems a challenging problem. The literature on this topic is constantly growing, but securing power systems against cyber-attacks is still an open problem \cite{RJCCPS},\cite{mo2012cyber},\cite{khurana2010smart}. 

In order to ensure secure operation of the power system, it is important to design and implement attack detection, which enables the system to mitigate any malicious intrusions. In general, this task is not trivial, since malicious attackers can be very resourceful, have detailed knowledge of the power system, and therefore launch highly effective and deceptive attacks. Static attack detectors do not consider system dynamics, but only the outputs of the system, which they check for consistency at every time step\cite{KosutStatic2010},\cite{Liu2009stat}. However, in \cite{pasqualetti2013attack} it is shown that attacks undetectable by static detectors can be constructed by the attacker. Stealthy attacks \cite{XieMoSin}, replay attacks \cite{mo2009secure}, and zero-dynamics attacks \cite{teixeira2012attack} are all examples of how an adversary can exploit knowledge of the system to launch attacks that can evade detection by the existing systems in power system control centers, namely State Estimation (SE) and static fault detection via Bad Data Detection (BDD). These attacks target SCADA (Supervisory Control and Data Acquisition), or more precisely the measurements of the system. Additionally, they don't affect the physical power system directly, but only through control and operating decisions based on wrong information. However, no dynamic detector can counter these attacks either, as they change the output of the system in a way that the output could also result from normal system behavior. References \cite{sridhar2014model}, \cite{Tabuada2016}, \cite{Fawzi2014} all provide dynamic attack detection algorithms. But in order to detect stealthy attacks, another class of detectors is introduced, namely active attack detectors. In contrast to passive detectors, active detectors perturb the system either through topology changes, or by injecting random signals into the network, in order to expose stealthy attacks.

One recent approach to active defense introduces an additional random signal, or "watermark", to the control signal as a form of authentication \cite{huang2018online},\cite{satchidanandan2017dynamic},\cite{mo2009secure}. In normal operation, this watermark should also be present in the measurement signal, so it's absence suggests that the system has been tampered with.  This is a good defense strategy, especially against replay attacks, but it is not effective in the case the attacker has extensive knowledge of the system model and the watermark. Another approach is to reveal the stealthy attacks by modifying the system's structure. Specifically, new measurements can be added incrementally to reveal stealthy zero-dynamics attacks \cite{teixeira2012revealing}. Even though this strategy effectively increases the robustness of the system, it is only successful for attacks that are constructed off-line, and once they are launched, the adversary can't gain new information about the changes in the structure of the system. Coding sensor outputs \cite{miao2014coding},\cite{rhouma2016coding} is an economical way of detecting stealthy FDI attacks when the attacker knows the system model without the coding scheme. However, similarly to the previous approaches, this strategy fails when the attacker has extensive knowledge of the system. On the other hand, in \cite{yuan2015security} the authors assume that the attacker without previous knowledge can first identify the system model by observing the control and measurement signals. Then they provide a controller design method which renders the system unidentifiable, but with a performance trade-off.

To address these issues, in this paper we make the following contribution. We propose a moving-target detection algorithm, that takes advantage of natural variations in power system dynamics, and varies its strategy accordingly in order to detect stealthy FDI attacks. The main difference in our approach, compared to the ones stated above, is that we don't assume that our detection strategy is always unknown to the adversary. Even if the adversary has complete knowledge of the current detection policy, system parameters and its structure at one point in time, the natural fluctuations of the system will render that knowledge unusable as the time goes by. The operating point of power systems, and therefore its dynamic behavior, varies in time since the generation has to continually change to satisfy the current demand.  With that in mind, our main goal is to quickly detect FDI attacks, before they have a chance to cause potentially devastating damage to the system. With our method, each incoming measurement is verified before being used as a feedback signal in any control process. We perform this verification in two stages. First, we cluster all the system measurements according to their dynamic response to natural load variations during normal operation. Next, we verify each incoming measurement by running a series of similarity checks with other measurements from the same cluster. 

\textit{Notation}: Let $\mathbb{R}$ and $\mathbb{N}$ denote the sets of real and natural numbers, respectively. Let $\mathcal{I}_k$ denote a set of integers, and $|\mathcal{I}_k|$ its cardinality. Then, $e^n_{\mathcal{I}_k} \in \mathbb{R}^{n\times |\mathcal{I}_k|}$ is a matrix composed of column vectors of the identity matrix ${I_n \in \mathbb{R}^{n \times n}}$ corresponding to the index set $\mathcal{I}_k$. Given a matrix $M \in \mathbb{R}^{n \times n}$, we denote its transpose by $M^T$, and its trace by $\tr\{M\}$. We use $\|M\|_F$ to denote the Frobenius norm of $M$, defined as ${\|M\|_F=\sqrt{\tr\{M^TM\}}}$. Furthermore, $M$ is said to be \textit{Hurwitz}  if every eigenvalue of $M$ has strictly negative real part. Matrix $M$ is said to be semistable iff the zero eigenvalues of $M$ are semisimple (geometric multiplicity of each eigenvalue coincides with its algebraic multiplicity), and all the other eigenvalues have a negative real part. Given a stable proper transfer function of a dynamical system $g(s)$, $\|g(s)\|_{\mathcal{H}_2}$ is the $\mathcal{H}_2$-norm of the system, defined as the energy of the system's impulse response. We use $\diag(v)$ to represent a diagonal matrix having a vector $v$ on its diagonal.

 The remainder of this paper is structured as follows. In Section~\ref{sec:problem} we formulate the problem of detecting FDI cyber-attacks. Then, in Section~\ref{sec:modeling} we introduce relevant power system component models and derive the dynamics of the interconnected power system. In Section~\ref{sec:detection}, we present the main contribution of this paper, the moving-target attack detection method. Section~\ref{sec:simulations} demonstrates the efficiency of our proposed method through a numerical example on the IEEE RTS 24-bus power system. Finally, in Section~\ref{sec:conclusion} we give some concluding remarks.

\section{Problem Formulation}
\label{sec:problem}
In this paper, we are concerned with detection of False Data Injection (FDI) attacks that target measurements, without affecting the physical behavior of the power system directly, but through incorrect control actions. Figure \ref{fig:blkdiag} depicts the block diagram representation of the system we consider. The power system block represents various components, such as generators, loads, transmission lines, etc. Two types of inputs enter the power system block: \textit{controlled inputs} $u(t)$ enforced by the actuators with the goal to stabilize and regulate the power grid to the desired operating point $x^0$, and \textit{uncontrolled inputs} or \textit{disturbances} $d(t)$ that represent changes in the environment that cannot be controlled. In our context, $d(t)$ is the power consumed by the loads, which directly influences control decisions and operation of the grid. More specifically, in a linearized system, $d(t)$ is the unpredictable variation of load around the forecasted value. It is important to note that this kind of disturbance does not act as noise in the system, since control feedback loop is designed to drive the actuators to balance these changes in load. Thus, we assume the physical system follows:
\begin{equation}
	\dot{x}(t)=Ax(t)+B\tilde{u}(t)+Gd(t)
	\label{eq:phys_dyn}
\end{equation}
where the states of system are denoted $x\in\mathbb{R}^n$, the disturbance is $d\in\mathbb{R}^m$, the control signal issued by the Control Center (CC) is $\tilde{u}\in\mathbb{R}^p$.
\begin{figure}[t]
	\begin{center}
		\includegraphics[width=0.48\textwidth]{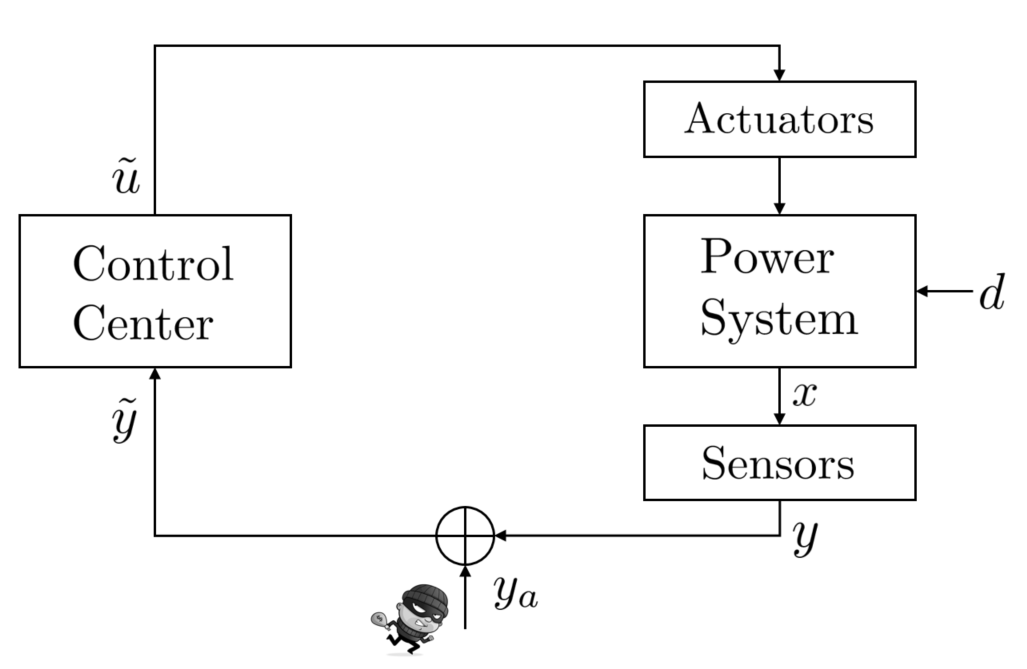}
		\caption{Block diagram of attacked power system in \eqref{eq:cl_att_lin_sys}. The attacker injects signal $y_a$ into measurements $y$ in order to manipulate the system.}
		\label{fig:blkdiag}
	\end{center}
\end{figure}
 A large network of field sensors is deployed to monitor the operation of the power system in \eqref{eq:phys_dyn}. In normal operation, the CC uses the measurements provided by these sensors to estimate the current demand, and then dispatch the generators accordingly. The actuators then adjust the power outputs of the generators to meet the new demand. However, a malicious attacker can negatively impact the system by manipulating the measurements, which is represented with the added signal $y_a(t)$ in Figure~\ref{fig:blkdiag}:
 \begin{align}
 \tilde{y}(t)&=y(t)+y_a(t)\\
 y(t)&=Cx(t)
 \label{eq:outputs}
 \end{align}
 where $\tilde{y} \in \mathbb{R}^l$ are measurements received by the CC. 
  Thus, a potentially manipulated measurement signal $\tilde{y}$ reaches the CC, which can then issue a potentially incorrect control signal $\tilde{u}$ to the power system actuators:
  \begin{equation}
  \tilde{u}(t)=K\tilde{y}=K[y(t)+y_a(t)]
  \end{equation}
 The control signal can then be decomposed as $\tilde{u}(t)=u(t)+u_a(t)$, where $u(t)=Ky(t)$ and $u_a(t)=Ky_a(t)$. Thus, the attacked system can be rewritten in closed-loop as:
\begin{equation}
\Sigma_a:\begin{dcases}
 \dot{x}(t)&=\mathcal{A}x(t)+\mathcal{B}y_a(t)+Gd(t) \\
\tilde{y}(t)&=Cx(t)+y_a(t)
\end{dcases}
\label{eq:cl_att_lin_sys}
\end{equation}
where $\mathcal{A}=A+BKC$ is the closed-loop system matrix, and $\mathcal{B}=BK$.

Note that both signals $d(t)$ and $y_a(t)$ are unknown. Now, the challenge we tackle in this paper emerges: distinguish whether a received measurement $\tilde{y}$ corresponds to physical changes in the system created by $d(t)$, or to malicious injection of $y_a(t)$ into the communication channel with intent to manipulate the control process. This is indeed a challenge when the attacks are stealthy, i.e. designed to satisfy the dynamics of the underlying system, so to be undetectable by traditional methods. This kind of disturbance, that does not represent noise, but physical changes in the system that are not predictable, can not be neglected when designing an FDI detection filter. We will demonstrate this in the next section. 
\subsection{Dynamic Detection Filter Based on Linear Observers}
In this section, we will explain why dynamic attack detectors based on linear observers may not be suitable for attack-detection in systems that are described by \eqref{eq:phys_dyn}, especially when the unknown disturbance $d(t)$ is not taken into account in the observer design process. Under normal conditions (no cyber-attack injected into the system, $\tilde{u}(t)=u(t)$), a linear observer is designed for the system in \eqref{eq:phys_dyn}, to compute the state estimate $\hat{x}(t)$ from the received measurements:
\begin{equation}
	\begin{aligned}
	\dot{\hat{x}}(t)&=\mathcal{A}\hat{x}(t)+L(\hat{y}(t)-y(t))\\
	&=(\mathcal{A}-LC)\hat{x}(t)+Ly(t)\\
	\hat{y}(t)&=C\hat{x}(t)
	\end{aligned}
	\label{eq:typ_observer}
\end{equation}
The error system can then be defined as $e(t):=\hat{x}(t)-x(t)$, and its dynamics:
\begin{equation}
	\dot{e}(t)=(\mathcal{A}-LC)e(t)-Gd(t)
\end{equation}
It is well known that, if $(\mathcal{A},C)$ is observable, $L$ can be chosen such that $(\mathcal{A}-LC)$ is Hurwitz, making the error system exponentially stable~\cite{ChenLinSysBook} in absence of the disturbance $d(t)$, i.e when $d(t)=0$. Finally, the residual $r_c(t):=\hat{y}(t)-y(t)$, where $\hat{y}$ is defined in \eqref{eq:typ_observer}, can be expressed as
\begin{equation*}
r_c(t)=Ce(t)
\end{equation*}
 Note that residual $r_c$ also converges to zero as $e(t) \rightarrow 0$ when $d(t)=0$. Therefore, a typical observer of this form will successfully estimate the state of the system under normal conditions only when the system does not experience a disturbance. Under certain assumptions on the disturbance, i.e. if it assumed to be white noise, an observer can be designed so that the residual converges to zero in expectation, i.e. $\mathbb{E}[r_c(t)]=0$. In that case, an accurate system estimate can be produced even in presence of a stochastic disturbance. In previous work, this is considered the case, and any unknown disturbance $d(t)$ that enters the system is neglected. Under this assumption, the observer error dynamics for a system under attack can be described by:
 \begin{equation*}
 \dot{e}(t)=(\mathcal{A}-LC)e(t)+(L-\mathcal{B})y_a(t)
 \end{equation*}
 and the residual by 
 $$r_c(t)=Ce(t)-y_a(t)$$
 Thus, one can use this residual to successfully detect cyber-attacks, as the error $e(t)$ and residual $r_c(t)$ will both converge to nonzero values only when a cyber-attack is present in the system, i.e. $y_a\neq 0$.
 On the other hand, in presence of the disturbance $d(t)\neq 0$, both the error and the residual will converge to nonzero values, $e\rightarrow e^*\neq 0$ and $r_c\rightarrow r_c^*\neq 0$, even when the system is not under attack. Specifically, $e^*(t)=(\mathcal{A}-LC)^{-1}Gd(t)$ and $r_c^*(t)=Ce^*(t)$. Thus, this observer will have a nonzero residual even in absence of cyber-attacks. If we also consider that the system is under attack, the error dynamics become:
 \begin{equation*}
 \begin{aligned}
 	\dot{e}(t)&=(\mathcal{A}-LC)e(t)+[(L-\mathcal{B})y_a(t)-Gd(t)]\\
 	r_c(t)&=Ce(t)-y_a(t)
 	\end{aligned}
 \end{equation*}
 Note that both the disturbance $d(t)$ and attack signal $y_a(t)$ enter the dynamics of the error system and residual. Therefore, using a detection filter based on such observers to detect FDI attacks on measurements  may be ineffective for several reasons. First, the effect of the disturbance $d(t)$ in the error dynamics may frequently trigger false alarms. Second, it will be ineffective against stealthy attacks, since the signals $d(t)$ and $y_a(t)$ may not be distinguishable. 
 Finally, it is clear that dynamic attack detectors, that do not account for unknown disturbances $d(t)$, cannot be used to detect stealthy cyber-attacks, as the residual $r_c$ is mutually dependent on both the disturbance and the attack signal, and it is nonzero even when the system is attack-free. In the rest of this paper, we will show how one can take advantage of the system's structure to expose stealthy attacks in a dynamic fashion by designing a new residual and a corresponding detection algorithm.
\subsection{Proposed Clustering-based Detection Filter}
In order to be able to discriminate between effects of the signals $d(t)$ and $y_a(t)$, we will first consider the system in (\ref{eq:cl_att_lin_sys}) in absence of attacks, i.e. for $y_a(t) \equiv 0$:
\begin{equation}
\Sigma:\begin{dcases}
\dot{x}(t)&=\mathcal{A}x(t)+Gd(t) \\
y(t)&=Cx(t)
\end{dcases}
\label{eq:cl_lin_sys}
\end{equation}
These equations represent the normal behavior of the system. To quantify the effects of disturbances on measurement signals, we define the clusters $\mathcal{I}_k$ as subsets of measurements that have a similar dynamic response to $d(t)$. More specifically, measurements $i,j$ belonging to the cluster $\mathcal{I}_k$ are approximately proportional $a_iy_i(t) \approx a_jy_j(t) \approx \dots \approx z^{(k)}(t)$, where $a_i, a_j, \dots$ are some constant coefficients. This relation of measurements within the same cluster can also be written as:
\begin{subequations}
\begin{align}
	&\hat{y}^{(k)}(t)=\begin{bmatrix}
	p_1^{(k)} \\  \vdots \\ p_{|\mathcal{I}_k|}^{(k)}
	\end{bmatrix} z^{(k)}(t) \label{eq:similar_meas}\\
	\text{such that} \quad &\|y^{(k)}(t)-\hat{y}^{(k)}(t)\|\leq \theta, \quad \theta \geq 0 \label{eq:similar_meas_cond}
	\end{align}
\end{subequations}
where $p_i=a_i^{-1}$ for all $i \in \mathcal{I}_k$, and $y^{(k)}$ is a subset of measurements $y$ belonging to cluster~$\mathcal{I}_k$, $y^{(k)}=(e^n_{\mathcal{I}_k})^Ty$, and all its elements are approximately proportional to a single scalar variable $z^{(k)}$. Parameter $\theta$ is a design parameter, which will be discussed in later sections.
Thus, we use terms \textit{cluster} and \textit{clustering} in the sense of model order reduction method in this work. We aim to estimate the full system state not only based on the received measurements and full system model, but also using the reduced order model which contains and encodes more detailed information (e.g. as in \eqref{eq:similar_meas}) that can be used to detect stealthy attacks. Specifically, we use the knowledge of the fact that incoming measurements $\tilde{y}^{(k)}(t)=y^{(k)}(t)+y_a^{(k)}(t)$ belonging to cluster $k$ also have the property in \eqref{eq:similar_meas} only if $y_a^{(k)}\equiv 0$.
Therefore, we define the set of residuals $r_{i,j}(t)$ that exploit this property as:
\begin{equation}
	r_{i,j}(t)= ||p_j\tilde{y}_i(t)-p_i\tilde{y}_j(t)||, \quad \forall i,j \in \mathcal{I}_k
	\label{eq:cbdf_residual}
\end{equation}
The residuals $r_{i,j} $ defined above will have a value larger than some threshold $\varepsilon$ only in presence of FDI attacks. The appropriate choice for the threshold $\varepsilon$ will be discussed in later sections.
In the following sections, we provide a method for finding the boundaries of clusters used for attack detection and design a moving-target FDI attack detection algorithm.

\section{Cyber-attack Detection Methodology}
\label{sec:detection}
In order to address the problem of detecting FDI attacks in a large power system network, we propose a methodology that consists of two steps. In the first step, we group the state measurements $y(t)$ into clusters, based on similarity of their dynamic response to input $d(t)$ in normal operating conditions. Then, in step two, we identify the attacked measurements through a series of consistency checks within each cluster. This procedure can be written compactly as:
\begin{itemize}
	\item group the indices $i$ of all measurements $y_i(t)$ with similar dynamic responses into index sets ${\mathcal{I}_k, \; k\in \{1,\dots K\}}$ 
	\item for a given threshold $\varepsilon >0$, check all incoming measurements $\tilde{y}_i(t)$ for consistency. If there exist $i,j\in \mathcal{I}_k$ such that following condition is satisfied 
	\begin{equation*}
		r_{i,j}(t)= \|p_j\tilde{y}_i(t)-p_i\tilde{y}_j(t)\| \geq \varepsilon
	\end{equation*}
	 flag measurements $\tilde{y}_i$ and $\tilde{y}_j$ as attacked measurements.
\end{itemize}
In the following subsections, we will introduce each of these steps in detail.
\subsection{Measurement clustering}
\label{sec:MOR}
In this section, we will introduce an aggregation method, inspired by \cite{ISHIZAKIpos_dir},\cite{Ishizaki20115019}, for clustering measurements according to their dynamic response to the external input $d(t)$. This method will form the basis of our methodology for detecting cyber-attacks in power systems. We begin by introducing the definition of a cluster.
\begin{definition}
	Let $\mathbb{L}=\{1,\dots, l\}$ be the set of measurement indices, and $\mathbb{K}=\{1,\dots,K\}$ the set of cluster indices. Then, clusters $\mathcal{I}_k$, $k\in \mathbb{K}$, are defined as disjoint subsets of $\mathbb{L}$, that cover all the elements in $\mathbb{L}$, i.e. $\bigcup_{k\in \mathbb{K}}\mathcal{I}_k=\mathbb{L}$. 
\end{definition}
Further, clustering coefficients corresponding to $\mathcal{I}_k$ are defined as $p_k\in \mathbb{R}^{1\times |\mathcal{I}_k|}$, such that $\|p_k\|=1$, and
\begin{equation}
p_k=\frac{(e^n_{\mathcal{I}_k})^T v_{max}}{\|(e^n_{\mathcal{I}_k})^T v_{max}\|}, \quad k\in \mathbb{K}
\label{eq:coefficients}
\end{equation}
where $v_{max}$ is the left eigenvector corresponding to the largest eigenvalue of $A$.
With this definition in mind, we aim to partition the set $\mathbb{L}$ into clusters $\mathcal{I}_k$ such that
\begin{equation}
p_jg_i(s)=p_ig_j(s), \quad \forall i,j \in \mathcal{I}_k
\end{equation}
where $g_i$ is the $i$-th element of $g(s)=C(sI_n-\mathcal{A})^{-1}G$, an input-output transfer matrix of the system in \eqref{eq:cl_lin_sys}. The measurements $i$ and $j$ belonging to the same cluster $k$ have a proportional, or in some cases identical, response to the input $d(t)$. In that sense, the proportionality of transfer functions $g_i$ and $g_j$ can also be expressed as proportionality to some scalar function $\bar{g}$ corresponding to cluster $k$. Therefore, we can define a condition for cluster formation in a compact way as follows.
\begin{definition}
	A set of measurements $\{y_i\}$ should form a cluster $\mathcal{I}_k$ if there exists a scalar function $\bar{g}(s)$ such that:
	\begin{equation}
	(e_{\mathcal{I}_k}^n)^Tg(s)=p_k^T\bar{g}(s)
	\label{eq:equal_gs}
	\end{equation} 
	where $p_k$ is defined as in \eqref{eq:coefficients}.
\end{definition}
This definition provides intuition on the meaning of clustering in our application, but is not practical for designing a procedure that would form such clusters, as it is not practical to perform similarity checks on functions. To get around this problem, we will derive a matrix-based condition equivalent to \eqref{eq:equal_gs} that will be more straightforward to check. For that reason, we will derive an equivalent condition for cluster formation that is more practical to check, based on this definition of similarity and the notion of reachability. To that end, we first derive reachability Gramian of a semistable system \eqref{eq:cl_lin_sys}. The reachability Gramian is defined as
\begin{equation}
	W_c=\int_{0}^{\infty}e^{\mathcal{A}t}GG^Te^{\mathcal{A}^Tt}dt
	\label{eq:ctrbGram}
\end{equation}
When $\mathcal{A}$ is Hurwitz, the above integral converges, and $W_c$ can also be found as a solution of the Lyapunov equation $\mathcal{A}W_c+W_c\mathcal{A}^T+GG^T=0$
 
 However, in power systems, the system matrix $\mathcal{A}$ has an inherent structural singularity, as a direct consequence of power conservation law. Due to semistability of the system matrix $\mathcal{A}$, the integral in \eqref{eq:ctrbGram} may not converge. To compute the reachability Gramian of a semistable system, we first consider the decomposition of $\mathcal{A}$ where $0=\lambda_1 > \lambda_2 \geq \dots \geq \lambda_n$
 \begin{equation*}
	\mathcal{A}= U\Lambda U^{-1}= [u_{max} \quad \bar{U}]\begin{bmatrix} 0 &  \\ & \bar{\Lambda}	\end{bmatrix} \begin{bmatrix} v_{max}^T\\ \\\bar{V}^T
	\end{bmatrix}
 \end{equation*} 
 where $u_{max}$ and $v_{max}$ are the right and left eigenvectors corresponding to the largest eigenvalue $\lambda_1=0$, and $\bar{\Lambda}$ is diagonal and Hurwitz. Let $\bar{\mathcal{A}}=\overline{V}^T\mathcal{A}\overline{U}$ and $\overline{G}=\overline{V}^TG$, defined as the stable subspace of $\Sigma$. Then, the reachability Gramian of the stable subspace is the solution of $\bar{\mathcal{A}}\;\overline{W}_c+\overline{W}_c\bar{\mathcal{A}}^T+\overline{G}\overline{G}^T=0$. Substituting $\bar{\mathcal{A}}$ and $\overline{G}$ into \eqref{eq:ctrbGram} yields
 \begin{align*}
 	\overline{W}_c&=\int_{0}^{\infty}e^{\bar{\mathcal{A}}t}\overline{G}\overline{G}^Te^{\bar{\mathcal{A}}^Tt}dt \\
 	&=\int_{0}^{\infty}\overline{V}^Te^{\mathcal{A}t}\bar{U}\overline{V}^TGG^T\overline{V}\bar{U}^Te^{\mathcal{A}^Tt}\overline{V}dt \\
 	&=\overline{V}^TW_c\overline{V}
 \end{align*}
 and 
 \begin{equation}
	W_c=\overline{V}\overline{W}_c\overline{V}^T
	\label{eq:semisGram}
 \end{equation} is the reachability Gramian of the semistable system $\Sigma$ and contains information on the degree of reachability of states with respect to the input $d(t)$. In the following theorem we show that the condition in \eqref{eq:equal_gs} is equivalent to linear dependence of rows of a matrix $\Phi$, where linearly dependent rows are expressed as proportional to some constant vector $\bar{\phi}$.
 \begin{theorem}
	Consider the reachability Gramian $W_c$ in \eqref{eq:semisGram} of the semistable system $\Sigma$ in \eqref{eq:cl_lin_sys}. Furthermore, let the Cholesky factorization of $W_c$ be given by ${W_c=W_LW_L^T}$, and ${\Phi=CW_L}$. Then, the condition in \eqref{eq:equal_gs} is equivalent to 
	\begin{equation}
		(e_{\mathcal{I}_k}^n)^T\Phi=p_k^T\bar{\phi}	
		\label{eq:equal_ws}
	\end{equation}
	where $\bar{\phi}\in\mathbb{R}^{1\times n}$ is a constant vector.
\end{theorem}
\begin{proof}
	In order for \eqref{eq:equal_gs} to hold, for each $i,j \in \mathcal{I}_k$ it must hold that
	$$p_j\|g_i(s)\|_{\mathcal{H}_2}=p_i\|g_j(s)\|_{\mathcal{H}_2}.$$
	Similarly, \eqref{eq:equal_ws} is equivalent to
	$$p_j\|\Phi_i\|=p_i\|\Phi_j\|$$
	where $\Phi_i$ is the $i$th row of the matrix $\Phi$.
	The $\mathcal{H}_2$-norm of a linear system can be computed as the $\mathcal{L}_2$-norm of its impulse response $h(t)$:
	$$\|g(s)\|_{\mathcal{H}_2}^2=\|h(t)\|_{2}^2=\tr\left\{C\int_{0}^{\infty}e^{\mathcal{A}t}GG^Te^{\mathcal{A}^Tt}dtC^T\right\}$$
	Plugging in \eqref{eq:semisGram}, we have
	$$\|h(t)\|_{2}^2=\tr\left\{C\overline{V}\left[\int_{0}^{\infty}e^{\bar{\mathcal{A}}t}\overline{G}\overline{G}^Te^{\bar{\mathcal{A}}^Tt}dt\right]\overline{V}^TC^T\right\}$$ 
	For $\|h(t)\|_{2}^2$ to be finite, the integral above must be finite. Since $\bar{\mathcal{A}}$ and $\overline{G}$ are the stable subspace of $\Sigma$, we have
	$$\lim\limits_{t\rightarrow\infty}e^{\bar{\mathcal{A}}t}=0$$
	and  $\|h(t)\|_{2}^2$ is finite and equal to:
	\begin{align}
	\|g(s)\|_{\mathcal{H}_2}^2&=\|h(t)\|_{2}^2=\tr\{CW_cC^T\}=\tr\{CW_LW_L^TC^T\}= \nonumber \\ 
	&=\|CW_L\|_F=\|\Phi\|_F
	\end{align}	
	where $\|\cdot\|_F$ is a vector norm applied to each row of $\Phi$.
	Hence, \eqref{eq:equal_gs} is equivalent to \eqref{eq:equal_ws}.
\end{proof}

However, in real systems, the identity in \eqref{eq:equal_gs} is almost never the case. Therefore, we will relax the strict equality, and require 
\begin{equation}
	\|p_jg_i(s)-p_ig_j(s)\|_{\mathcal{H}_2} \leq \varepsilon, \quad \forall i,j \in \mathcal{I}_k
	\label{eq:clustering_condition}
\end{equation}
to hold for each cluster, or equivalently, we can check for linear dependence between rows of matrix $\Phi$:
\begin{equation}
\|p_j\Phi_i-p_i\Phi_j\| \leq \theta \quad \forall i,j \in \mathcal{I}_k
\label{eq:clust_cond}
\end{equation}
where $\theta >0$ and $\Phi_i$ is the $i$-th row of $\Phi$. Here, $\theta$ is a parameter that allows us to control the coarseness of clustering. In other words, it allows us to find outputs that have a "similar", instead of equal, response, which relaxes the condition \eqref{eq:equal_gs}. The smaller $\theta$ is, more accurate the clustering will be, but the clusters may contain very few measurements, which is not desirable for attack detection purposes. On the other hand, if $\theta$ is too large, the detection threshold in \eqref{eq:cbdf_residual} will have to be large as well to avoid false alarm, and stealthy attacks may not be detected. Since this trade-off in choice of $\theta$ is obvious, ideally, $\theta$ should be chosen as a smallest value for which each cluster contains at least two measurements.

Finally, we can introduce the measurement clustering algorithm defined above. Assume $k$ clusters have already been formed. First, we choose an index $i$ that hasn't already been assigned to any cluster, and add it to cluster $k+1$. Then, we choose another index $j$ that is not yet assigned to a cluster, and check condition \eqref{eq:clust_cond} for $i$ and $j$. If the condition is satisfied, we add $j$ to cluster $k+1$. We repeat this process until all measurements are assigned to a cluster. This procedure is summarized in the algorithm below.
 
\begin{algorithm}[H]
\label{clustering_alg}
  \begin{algorithmic}
  		\caption{Clustering algorithm}
  		\State \textbf{Initialize} cluster index $k=0$, and cluster set $\mathbb{K}=\emptyset$
  		\Repeat 
  			\State \textbf{Choose} measurement index $i \in \mathbb{L}$ that hasn't been 
  			\State \quad assigned to a cluster yet, and add it to cluster $\mathcal{I}_{k+1}$
  			\State \textbf{Set} $k \leftarrow k+1$, $\mathbb{K}\leftarrow \{\mathbb{K},k\}$
  			\State \textbf{Find} all $j \in \mathbb{L}$ that haven't been assigned to a cluster
  			\State \quad  yet and that satisfy \eqref{eq:clust_cond} and add them 
  			\State \quad to $\mathcal{I}_{k+1} \leftarrow \{\mathcal{I}_{k+1},j\}$
  		\Until {all measurements are assigned to a cluster, i.e. $\bigcup_{k\in \mathbb{K}}\mathcal{I}_k=\mathbb{L}$ }
  \end{algorithmic}
\end{algorithm}
 
 Next, we introduce the moving-target FDI attack detection method, based on this clustering procedure.
 
\subsection{Detection algorithm}
In the previous section we have derived the algorithm for cluster formation with regard to the linearized system in \eqref{eq:cl_lin_sys}. In this section, we will introduce our cyber-attack detection algorithm that leverages measurement clusters found using Algorithm 1. Two properties of this clustering method are key in our cyber-attack detection filter design. Firstly, we know that, once clustering is performed on the system in normal operating conditions, the outputs within the clusters will be approximately proportional to each other at all times $t$. That enables us to perform quick consistency checks to ensure the safety and reliability of the system. Secondly, the result of clustering will change over time as operating conditions change. On one hand, that means that the clustering will have to be performed periodically, but on the other hand, it gives our detection filter a very desirable property of "moving-target" behavior. That means our detection strategy will change over time, posing an additional difficulty to the attacker. In this section, we will derive a detection algorithm based on these two crucial properties.

In the analysis in Section~\ref{sec:MOR}, we have shown that clusters can be formed such that measurements $i,j$ within the cluster $\mathcal{I}_k$ are approximately proportional, i.e. $a_iy_i(t) \approx a_jy_j(t) \approx \dots \approx z^{(k)}(t)$, where $p_i=a_i^{-1}$ for all $i \in \mathcal{I}_k$. Then, we derived a matrix-based condition to find such clusters. Next, we will show that the clustering-based approximation $\hat{y}$ as in \eqref{eq:similar_meas} is a good approximation of $y$. In other words, we can approximate original system outputs ${y=Cx}$ by ${\hat{y}=\Pi^Tz}$, where $z=\Pi y=\Pi Cx$. The clustering matrix $\Pi \in \mathbb{R}^{K\times n}$ is defined as:
\begin{equation}
\Pi:=\Diag\{p_1,p_2,\dots, p_K\}E \in \mathbb{R}^{K \times n}
\label{eq:agg_matrix}
\end{equation}
where $p_i$ is defined in \eqref{eq:coefficients}, and $E$ is a permutation matrix defined as $E=[e^n_{\mathcal{I}_1}, \dots,e^n_{\mathcal{I}_K}]^T$. The input-output transfer matrix associated with $\hat{y}$ can then be defined as ${\hat{g}(s)=\Pi^T\Pi C\left(sI_n-\mathcal{A}\right)^{-1}G}$. The following theorem will show that $\hat{y}$ is a good approximation of $y$, and that it can be used in our cyber-attack detection methodology.
\begin{theorem}
	Consider a semistable linear system in \eqref{eq:cl_lin_sys}. Consider also a clustering-based approximation $\hat{y}$ obtained using the aggregation matrix $\Pi$. Then, the error system of the approximation $g_e(s)=g(s)-\hat{g}(s)$ is asymptotically stable.
\end{theorem}
\begin{proof}
	By definitions given in \eqref{eq:coefficients} and \eqref{eq:agg_matrix}, $\Pi$ is a unitary matrix, i.e. $\Pi\Pi^T=I_K$, and $\Pi^T\Pi$ is an orthogonal projection onto $\colspace(\Pi^T)$. Note also that, by definition, $v_{max} \in \colspace(\Pi^T)$. 
	We then define $\bar{\Pi}$ as an orthogonal complement of $\Pi$, such that $[\Pi^T \; \bar{\Pi}^T]^T$ is unitary, and $I-\Pi^T\Pi=\bar{\Pi}^T\bar{\Pi}$.
	Consider now the error system $g_e=g-\hat{g}$ of the approximation:
	\begin{equation}
	\begin{aligned}
	g_e(s)&=C(sI-\mathcal{A})^{-1}G-\Pi^T\Pi C(sI-\mathcal{A})^{-1}G= \\
	&=(I_n-\Pi^T\Pi)C(sI-\mathcal{A})^{-1}G=\bar{\Pi}^T\bar{\Pi}g(s)
	\end{aligned}
	\end{equation}
	Then $\Pi^T\Pi v_{max}=v_{max}$, or equivalently 
	$$\bar{\Pi}v_{max}=0.$$
	This implies that there is pole-zero cancellation in $\bar{\Pi}g(s)$ associated with the zero eigenvalue. Therefore, all poles of $\bar{\Pi}g(s)$ have negative real parts, and the error system $g_e$ is asymptotically stable.
\end{proof}

 Now, we introduce the moving-target detection algorithm based on dynamic clustering of system outputs. Firstly, the moving-target nature of our proposed method stems from the natural fluctuations occurring in power systems. As the underlying power system model is nonlinear, the system matrices $\mathcal{A}$ and $G$ are only valid around a certain operating point $x^0$, and we will denote them with $\mathcal{A}(x^0)$ and $G(x^0)$. Unit Commitment (UC) and Economic Dispatch (ED) are core processes in power system's operation, whose main function is to balance the current demand with sufficient generation in regular intervals (usually every 1 hour). Thus, as loading or topology of the system change, so will the operating point of the system  $x^0$, as well as the boundaries of the clusters. 

As a result, cluster boundaries have to be recomputed at every instance of ED, for the current operating conditions $x^0$. At first, this may sound inconvenient for the system operator, but the additional computations that are needed are negligible compared to the significant increase in difficulty of manipulating the system using the moving-target strategy. In other words, even if the attacker had absolute knowledge of the system at one time, our detection strategy will keep changing, and that knowledge will eventually become meaningless. 

Once the new operating point is received from the ED process, the linearized system matrices and clusters need to be recomputed. Then, at every time step of the control processes, the incoming measurements must be verified using residuals of the clustering-based detection filter given in \eqref{eq:cbdf_residual} before the control action is computed and performed. Verification is performed within each cluster using the following procedure.
\begin{itemize}
	\item For each cluster $\mathcal{I}_k, \; k \in \mathbb{K}$, run a sequence of similarity checks for all $i,j \in \mathcal{I}_k$:
	\begin{equation}
	r_{i,j}(t)= ||p_j\tilde{y}_i(t)-p_i\tilde{y}_j(t)||, \quad \forall i,j \in \mathcal{I}_k
	\label{eq:sim_check}
	\end{equation}
	\item All measurements that satisfy $r_{i,j}(t) \geq \epsilon$ are classified as attacked measurements.
\end{itemize}
 Now, we can state the moving-target FDI attack detection algorithm, based on dynamic clustering.
 
 \begin{algorithm}[H]
 \begin{algorithmic}
 	\caption{Clustering-based FDI attack Detection Method}
 	\Repeat
 	\State \textbf{Get} new operating conditions $x^0$
 	\State \textbf{Compute} matrices $\mathcal{A}(x^0)$, $G(x^0)$, and find cluster
 	\State  \quad sets $\mathcal{I}_k$ according to Algorithm 1
 	\For {every time-step of the control process}
 			\For {all measurements $y_i, i \in \mathcal{I}_k$}
 				\If {condition in \eqref{eq:sim_check} is satisfied}
 					\State $\rightarrow$ \textbf{FDI attack detected} 
 				\EndIf
 			\EndFor
 	\EndFor 
 	\Until {new ED interval}
 \end{algorithmic}
\end{algorithm}

Finally, we give a note on the computation effort involved in implementing the detection algorithm described above. As we explained, the clustering procedure needs to be repeated relatively infrequently, usually every 1 hour. Even so, in the large-scale setting, solving Lyapunov equations may be considered time consuming. However, the inherent sparsity of power systems, and a number of existing efficient iterative calculation methods for solving the Lyapunov equation, enable the clustering to be performed in a timely manner. The verification process of system measurements only consists of simple mathematical and logical operations, and is not computationally demanding. Therefore, our proposed FDI attack detection method can be implemented in real-time system operation.
\section{Numerical examples}
\label{sec:simulations}
We begin this section by providing necessary power system component models, and deriving the standard state space model of the interconnected power system in form given in \eqref{eq:cl_lin_sys}. Then, we demonstrate the performance of our method on the IEEE RTS 24-bus system.
\subsection{Power system modeling}
\label{sec:modeling}
In this section, we introduce the power system component models used to derive the system matrices in \eqref{eq:cl_lin_sys}. We model the loads as dynamic using the structure-preserving load model \cite{mi2010}, \cite{dh1981}, alongside the well-known generator model with governor control \cite{IlicZab}, for two main reasons. First, so that the sparsity of the power system topology is preserved. Second, and even more importantly, to account for cyber-attacks on loads such as smart meters, electric vehicles and other smart appliances \cite{RJCCPS}, that are exposed through their Internet connectivity. For these reasons, modeling loads when studying the impact of cyber-attacks on power systems becomes highly necessary.
\par Here we consider a power system with $n_G$ generators and $n_L$ loads, and denote the set of generator buses by $\mathcal{G}$, and the set of load buses by ${\mathcal{L}}$. The mechanical dynamics of generators and aggregate loads at the substation level are given by:
\begin{equation}
\begin{aligned}
\label{eq:mech_dyn}
J_i\dot{\omega}_i+D_i\omega_i&=P_{T,i}-P_i+e_{T,i}a_i,&& i\in\mathcal{G}\\ 
J_i\dot{\omega}_i+D_i\omega_i&=-P_i-L_i,&& i\in\mathcal{L}
\end{aligned}
\end{equation}
For each bus $i$, state variable $\omega_i$ denotes its frequency, $P_i$ the net real power injected into the network, and parameters $J_i$, $D_i$ the inertia and damping. At load buses $i \in \mathcal{L}$, $L_i$ is defined as actual mechanical power consumed by the load. At generator buses $i \in \mathcal{G}$, there are additional controller dynamics, namely, governor dynamics. States $P_{T,i}$ and $a_i$ denote the mechanical power of the generator and the turbine valve position respectively, and $e_{T,i}$ is a parameter of the turbine. The governor dynamics are given by:
\begin{equation}
\begin{aligned}
T_{u,i}\dot{P}_{T,i}&=-P_{T,i}+K_{t,i}a_i,&&i\in\mathcal{G} \\
T_{g,i}\dot{a}_i&=-r_ia_i-(\omega_i-\omega^{ref}),&&i\in\mathcal{G}
\end{aligned}
\label{eq:gov_dyn}
\end{equation}
The governor's and turbine's time constants are denoted by $T_{u,i}$ and $T_{g,i}$ while $K_{t,i}$ and $r_i$ are control gains. Finally, $\omega^{ref}$ is the frequency reference provided by the higher control layer. In order to derive the interconnected system, we treat $P_i$ as a coupling state variable whose dynamics can be obtained by differentiating the DC power flow equation, and expressed in matrix form as:
\begin{equation}
\begin{bmatrix} \dot{P}_G \\ -\dot{P}_L\end{bmatrix}=Y_{bus}\mathbf{\omega}
\label{eq:net_dyn}
\end{equation} 
where $P_G:=[P_i]_{i\in\mathcal{{G}}}$ and $P_L:=[P_i]_{i\in\mathcal{{L}}}$, and $Y_{bus}$ is the admittance matrix of a lossless transmission network which can be partitioned as:
\begin{align*}
Y_{bus}=\begin{bmatrix}Y_{GG} & Y_{GL}\\
Y_{LG} & Y_{LL}
\end{bmatrix} 
\end{align*}
Further, $\omega$ can also be partitioned as $\omega:=\begin{bmatrix}\omega_G & \omega_L \end{bmatrix}^T$ where ${\omega_G:=[\omega_i]_{i\in\mathcal{{G}}}}$ and ${\omega_L:=[\omega_i]_{i\in\mathcal{{L}}}}$. Finally, the linearized closed-loop power system model described by equations \eqref{eq:mech_dyn}, \eqref{eq:gov_dyn} and \eqref{eq:net_dyn} can be expressed in form \eqref{eq:cl_lin_sys}, with the state vector $x:=[\omega_G, \omega_L, P_G, P_L, P_T, a]^T$, and ${d=[L_i]_{i\in\mathcal{L}}}$.
The system matrices are given by:
\setlength{\arraycolsep}{2pt}
\begin{align*}
G&:=\begin{bmatrix}
\mathbf{0}_{n_G\times n_L} \\
-\mathbf{J}_L^{-1} \\
\mathbf{0}_{(3n_G+n_L)\times n_L}
\end{bmatrix}
\end{align*}
\begin{align*}
\mathcal{A}&:=\begin{bmatrix}
-\mathbf{J}_G^{-1}\mathbf{D}_G & \mathbf{0} & -\mathbf{J}_G^{-1} & \mathbf{0} & \mathbf{J}_G^{-1} & \mathbf{J}_G^{-1}\mathbf{e}_T \\
\mathbf{0} & -\mathbf{J}_L^{-1}\mathbf{D}_L & \mathbf{0} & \mathbf{J}_L^{-1} & \mathbf{0} & \mathbf{0} \\
\textbf{Y}_{GG} & \textbf{Y}_{GL} & \mathbf{0} & \mathbf{0} & \mathbf{0} & \mathbf{0}\\
-\textbf{Y}_{LG} & -\textbf{Y}_{LL} & \mathbf{0} & \mathbf{0} & \mathbf{0} & \mathbf{0}\\
\mathbf{0} & \mathbf{0} & \mathbf{0} & \mathbf{0} & -\mathbf{T}_u^{-1} & \mathbf{T}_u^{-1}\mathbf{K}_t\\
-\textbf{T}_g^{-1} & \mathbf{0} & \mathbf{0} & \mathbf{0} & \mathbf{0} & -\textbf{T}_g^{-1}\textbf{r}
\end{bmatrix}
\end{align*}
where:
\begin{align*}
	&\mathbf{J}_G:=\diag(\{J_i\}_{i\in\mathcal{{G}}}) & &\mathbf{K}_t:=\diag(\{K_{t,i}\}_{i\in\mathcal{{G}}})\\
	&\mathbf{J}_L:=\diag(\{J_i\}_{i\in\mathcal{{L}}}) & &\mathbf{T}_g:=\diag(\{T_{g,i}\}_{i\in\mathcal{{G}}})\\
	&\mathbf{D}_G:=\diag(\{D_i\}_{i\in\mathcal{{G}}}) & &\mathbf{T}_u:=\diag(\{T_{u,i}\}_{i\in\mathcal{{G}}})\\
	&\mathbf{D}_L:=\diag(\{D_i\}_{i\in\mathcal{{L}}}) & &\mathbf{r}:=\diag(\{r_{i}\}_{i\in\mathcal{{G}}})\\
	&\mathbf{e}_T:=\diag(\{e_{T,i}\}_{i\in\mathcal{{G}}}) 
\end{align*}
In the following sections, we present numerical simulation examples performed on the IEEE RTS 24-bus power system to illustrate the performance of our proposed cyber-attack detection method. 
\subsection{Test system and illustrative scenarios}
The IEEE RTS 24-bus system \cite{IEEE24bus} consists of 10 generators, equipped with governor control, and 14 loads. The interconnected system is modeled using equations \eqref{eq:mech_dyn}-\eqref{eq:net_dyn}, where the dimension of $x$ is 68. For this system, we will first consider two scenarios with different loading conditions, to demonstrate the detection algorithm as well as the moving-target defense strategy:
\begin{itemize}
\item Scenario 1 - the system is at high loading condition. From $t=0$ to $20$ s, the loading is nominal. At time $t=20$ s, load at bus 3 increases by $0.5$ p.u., and at time $t=200$ loading returns to nominal value.\\
\item Scenario 2 - the system is at low loading condition.  From $t=0$ to $20$ s, the loading is nominal. At time $t=20$ s, load at bus 3 increases by $0.5$ p.u., and at time $t=200$ loading returns to nominal value.
\end{itemize}
We use these two scenarios to demonstrate the clustering method introduced in Section~\ref{sec:MOR}, and how cluster boundaries change with the operating conditions. This is demonstrated in Table~\ref{tab:moving_target}, where the net real power injection of generator 8, $P_{G_8}$, and its mechanical power output $P_{T_8}$ are clustered with respective states of generators~3 and 10 under Scenario~1, and generators~4 and 5 under Scenario~2.
\begin{table}[h]
	\begin{center}
		\large
		\begin{tabular}{ r |l }
			 & Clustered states \\
			\hline
			Scenario 1 & $P_{G_3}$,\fbox{$P_{G_8}$}, $P_{G_{10}}$, $P_{T_3}$, \fbox{$P_{T_8}$}, $P_{T_{10}}$  \\ 
			Scenario 2 & $P_{G_4}$,  $P_{G_5}$, \fbox{${P}_{G_8}$}, $P_{T_4}$, $P_{T_5}$, \fbox{${P}_{T_8}$} \\ 
		\end{tabular}
	\end{center}
	\caption{$P_{G_8}$ and $P_{T_8}$ belong to different clusters as operating conditions change}
	\label{tab:moving_target}
\end{table}
Additionally, for the same $\theta=5e^{-3}$, the clustering procedure resulted in 21 clusters under Scenario~1, and 23 under Scenario~2. 

In Figure~\ref{fig:cluster7} we show the dynamic response of one of the clusters under Scenario~1. In this particular case we can also see the effect of the coarseness parameter $\theta$: for higher detection accuracy, we can decrease $\theta$ in which case the cluster in Figure~\ref{fig:cluster7} would split into two. Appropriate attack analysis and parameter tuning is necessary in general case, but we use reasonable values in a common attack scenario to demonstrate our method. 
\begin{figure}[h]
	\begin{center}
		\includegraphics[scale=0.5]{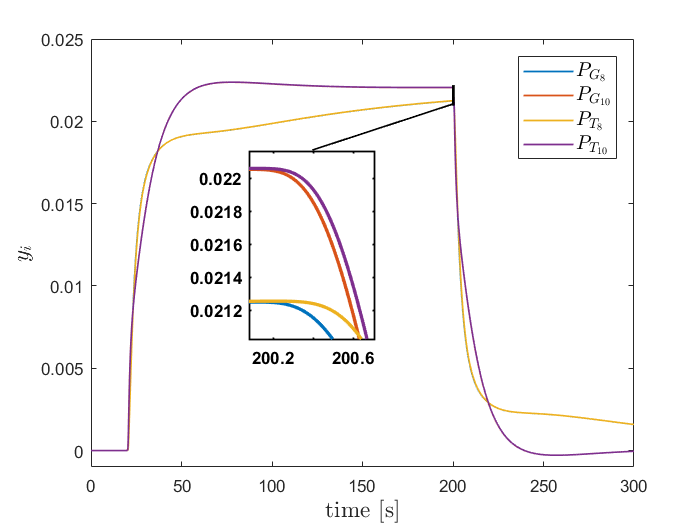}
		\caption{Dynamic response of measurements in one of the clusters under Scenario 1}
		\label{fig:cluster7}
	\end{center}
\end{figure}

We note that the nature of power systems highly influences the clustering procedure. In other words, measurements of states tend to group according to type:  generator net real power injection $P_{G,i}$ and mechanical power output $P_{T,i}$ measurements tend to be in the same group, while load net real power injections $P_{L,i}$ group together, and so do the frequencies $\omega_i$. Also, component and network parameters have an influence on cluster boundaries. Two generators with the same inertia, damping and controller gains will naturally have very similar dynamic response. Also, frequencies of components connected by a line with a large susceptance (i.e. small impedance), tend to have similar magnitude of oscillation.
\subsection{Implementation of the attack detection methodology and result analysis}
We consider Scenario~3 in the remainder of this section to demonstrate attack-detection capabilities of our proposed method:
\begin{itemize}
	\item Scenario 3: the system is at high loading condition. From $t=0$ to $20$ s, the loading is nominal. At time $t=20$ s, load at bus 3 increases by $0.5$ p.u., and at time $t=200$ loading returns to nominal value; at time $t=125$ s, a sequence of 6 scaling attacks are launched on the measurement $P_{G,8}$, each lasting 5 seconds, with total duration of the attack $T_a=55$ s
\end{itemize} 
where a scaling attack can be represented in terms of system in \eqref{eq:outputs} as $\tilde{y}=y+y_a$, where $y_a=k\cdot y$. We use scaling coefficient $k=0.1$, which corresponds to a 10 \% increase in value of $P_{G,8}$ at the time of the attack.
 
We use Scenario 3 to demonstrate the performance of our detection filter, both in presence and absence of measurement noise. In Figure~\ref{fig:attack} we consider the noiseless scenario, and compare the attacked measurement $P_{G,8}$ (middle plot) with other measurements belonging to the same cluster (top plot), to obtain the residual in the bottom plot, which only crosses the chosen detection threshold for each of the attacks. In Figure~\ref{fig:attack_noise} we consider the same attack scenario in presence of noise. Note that false positive alarms become very likely in this case, using the appropriate threshold designed for the deterministic scenario. If the noise parameters are known, additional statistical methods (e.g. hypothesis testing, etc.) may be employed to distinguish between noise and signal. 
In Figures~\ref{fig:noattack} (no measurement noise) and \ref{fig:noattack_noise} (with measurement noise) we show cluster measurements (top) and residuals (bottom) in absence of cyber-attacks. In noiseless scenario, the residual does not cross the detection threshold even when system conditions change, i.e. when there is a load disturbance in the system. In presence of noise, false positives are possible, and additional statistical methods can be employed to improve the performance of the attack detection filter.

\begin{figure*}[t]
	\centering
	\begin{subfigure}[t]{0.49\textwidth}
		\includegraphics[scale=0.39]{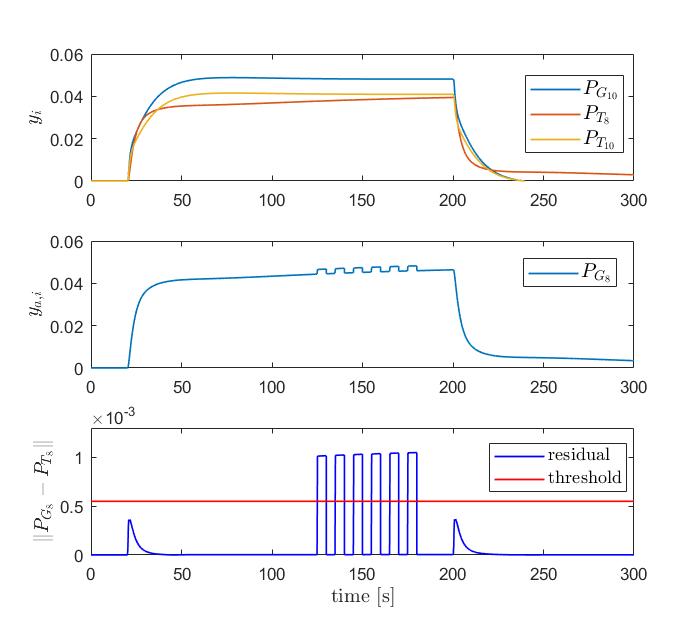}
		\caption{Top: Other cluster measurements; Middle: Measurement of $P_{G,8}$ under a scaling attack; Bottom: detection residual (blue) and threshold (red)}
		\label{fig:attack}
	\end{subfigure}~
	\begin{subfigure}[t]{0.49\textwidth}
		\includegraphics[scale=0.39]{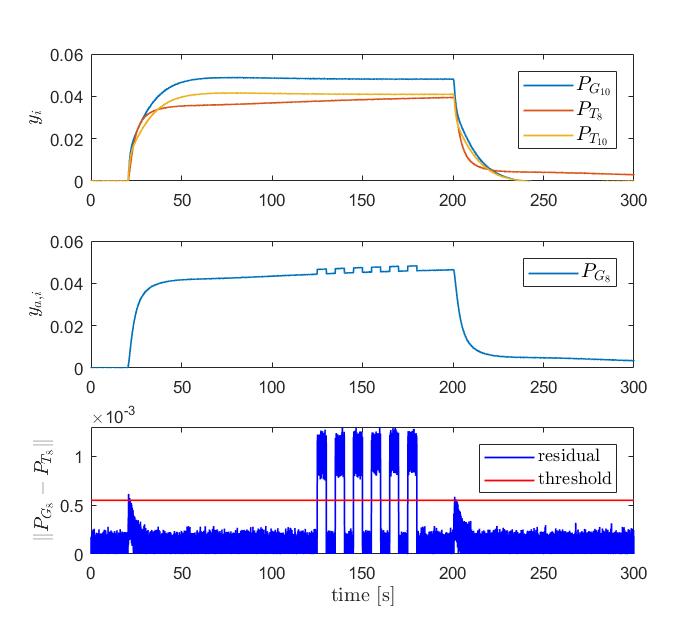}
		\caption{Top: Other cluster measurements (noisy); Middle: Measurement of $P_{G,8}$ (noisy) under a scaling attack; Bottom: detection residual (noisy) and threshold (red)}
		\label{fig:attack_noise}
	\end{subfigure}

	\begin{subfigure}[t]{0.49\textwidth}
		\includegraphics[scale=0.39]{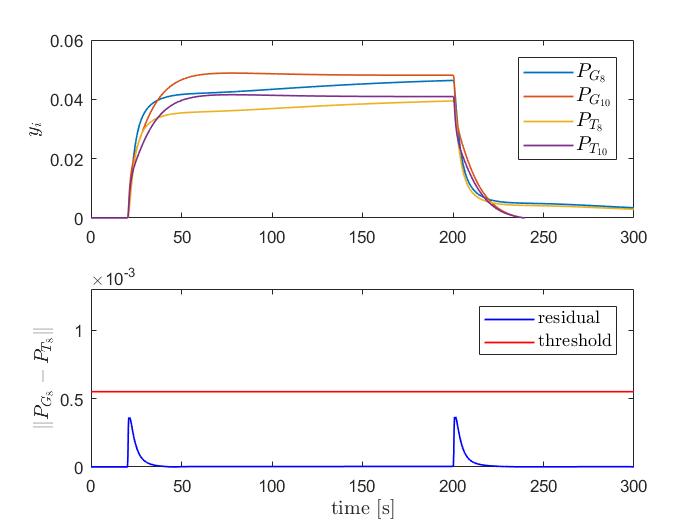}
		\caption{Top: Measurements in one of the clusters under Scenario 3, in absence of cyber-attacks; Bottom: residual in absence of cyber-attacks}
		\label{fig:noattack}
	\end{subfigure}~
	\begin{subfigure}[t]{0.49\textwidth}
		\includegraphics[scale=0.39]{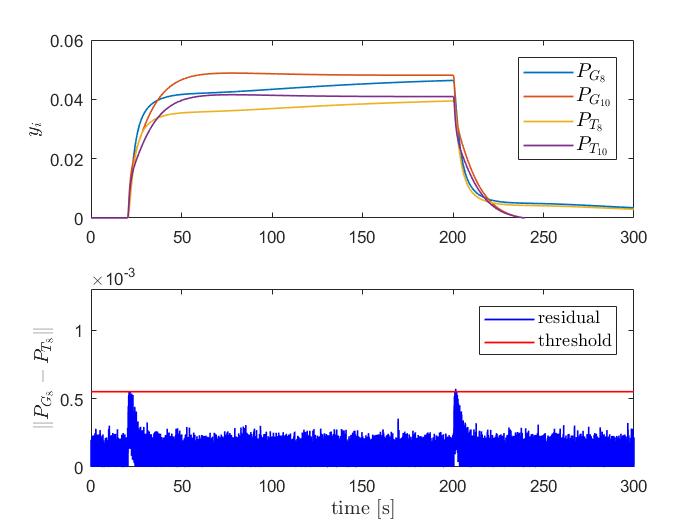}
		\caption{Top: Noisy measurements in one of the clusters under Scenario 3, in absence of cyber-attacks; Bottom: noisy residual in absence of cyber-attacks}
		\label{fig:noattack_noise}
	\end{subfigure}
\caption{Group of measurements belonging to one of the clusters under Scenario 3. Figures 3a and 3b show cases where one of the measurements in the cluster is under scaling attack, namely $P_{G_8}$, in noiseless and noisy setting, respectively. Figures 3c and 3d show measurements from the same cluster in absence of cyber-attacks, in noiseless and noisy setting,respectively.}
\end{figure*}

\section{Conclusion}
\label{sec:conclusion}
In this paper, we presented a moving-target defense algorithm that is based on dynamic clustering to detect false data injection attacks on power system measurements. Our strategy is dynamic and uses information about the system's changing operating point to define clusters of measurements that have similar dynamic response and then carries out detection through similarity checks on the measurements belonging to the same cluster. We numerically showed the performance of our proposed detection algorithm and its ability to successfully detect FDI attacks through an example on the IEEE 24-bus power system.

\bibliographystyle{unsrt}
\bibliography{Agg_journal_ver1}
\end{document}